\documentclass[letterpaper,USenglish,cleveref, autoref]{lipics-v2019}

\usepackage{amsmath}
\usepackage{amssymb}
\usepackage{stmaryrd}
\usepackage{booktabs}
\usepackage{mathtools}
\usepackage{tikz}
\usepackage[linguistics]{forest}
\usetikzlibrary{positioning,fit,arrows,backgrounds}
\usetikzlibrary{automata}

\newcommand{\set}[3]{{#1}_{#2}^{#3}}
\newcommand{\init}{v_{0}}
\newcommand{\bigo}{\mathcal{O}}
\newcommand{\plays}[1]{plays(#1)}
\newcommand{\pre}[1]{pref(#1)}

\newcommand{\Dj}[1]{\textcolor{red}{#1}}

\title{Playing Against Opponents With Limited Memory} 

\titlerunning{Playing Against Opponents With Limited Memory}

\author{Dhananjay Raju}{The University of Texas at Austin, USA}{draju@cs.utexas.edu}{}{}

\author{R\"udiger Ehlers}{Clausthal University of Technology, Germany}{ruediger.ehlers@tu-clausthal.de }{}{}{}

\author{Ufuk Topcu}{The University of Texas at Austin, USA}{utopcu@utexas.edu}{}{}

\authorrunning{Raju, Ehlers and Topcu}

\Copyright{Dhananjay Raju, R\"udiger Ehlers and Ufuk Topcu}



\category{Submitted to ICALP 2020}

\relatedversion{}

\supplement{}

\acknowledgements{}

\nolinenumbers 

\hideLIPIcs  

\EventEditors{}
\EventNoEds{0}
\EventLongTitle{}
\EventShortTitle{}
\EventAcronym{}
\EventYear{}
\EventDate{}
\EventLocation{}
\EventLogo{}
\SeriesVolume{}
\ArticleNo{}

\begin{document}

\maketitle
\begin{abstract}
We study \emph{partial-information} two-player turn-based games on graphs with omega-regular objectives, when the partial-information player has \emph{limited memory}. 
Such games are a natural formalization for reactive synthesis when the environment player is not genuinely adversarial to the system player. The environment player has goals of its own, but the exact goal of the environment player is unknown to the system player. We prove that the problem of determining the existence of a winning strategy for the system player is PSPACE-hard for reachability, safety, and parity objectives. Moreover, when the environment player is memoryless, the problem is PSPACE-complete. However, it is simpler to decide if the environment player has a winning strategy; it is only NP-complete.  Additionally, we construct a game where the the partial-information player needs at least $\mathcal{O}(\sqrt{n})$ bits of memory to retain winning strategies in a game of size $\mathcal{O}(n)$.
\end{abstract}

\clearpage
\section{Introduction}

Reactive synthesis is the process of computing correct-by-construction implementations for reactive systems from their specifications. The  attractiveness of leaving the implementation of a system to a synthesis algorithm led to the application of reactive synthesis in a diverse set of domains such as hardware circuits, graphical user interfaces, and high-level robotic mission planning \cite{Harel1985,doyen_raskin_2011,Bloem2018}.

The core difficulty in reactive synthesis is the computation of how parts of a specification can be satisfied without the knowledge of the future input to the system. The reactive synthesis problem is typically reduced to solving a game between two players: a player corresponding to the system and a player corresponding to the environment. The specification in the synthesis problem is encoded through the winning condition for the system player.  There exists a solution to the synthesis problem if and only if the system player has a strategy to win this game.

Most specifications are \emph{realizable} only under certain assumptions on the behavior of the environment. These assumptions have to also be captured in the specification. For instance, an elevator can only move between floors if its doors are closed. If the doors are blocked all of the time, a specification to eventually reach another floor cannot be fulfilled. Correspondingly, making \emph{assumptions} about the environment player into the games constructed for synthesis is necessary to solve the synthesis problem. In practice, specifying these assumptions precisely is difficult.  
Furthermore, requiring the system player to satisfy the specification only if the assumptions hold creates an incentive for the system player to violate the assumptions actively. Raskin et al.~\cite{Hunter2017Feb}~have for instance addressed this problem by computing \emph{strategy profiles} for the environment and system players so that none of them has an incentive to deviate from their strategy.

Synthesis based on strategy profiles relies on the often unreasonable assumption that the precise goals of the environment player is available to the system player. Typically, only an approximate understanding of the potential environment behavior is available. As such these assumptions are frequently insufficient to guarantee the existence of a winning strategy for the system player.
Consider an example from the robotics domain, where a common problem is that a robot should operate in a workspace shared with a human without collision. The goal of the human, i.e.\ the environment player, is usually not known, and the human could indeed \emph{always} move in a way that the robot's path is blocked. In this case, by using the conventional game setting, we cannot compute \emph{any} winning strategy for the system player, and hence synthesizing a controller for the robot fails.
For successful synthesis, we need to restrict the environment's behavior without making assumptions about the environment's goals.

In this paper, we present a new class of games in order to address this seemingly self-contradicting requirement:  we restrict the abilities of the environment player without any prior knowledge on its goals.
Our starting point is the common idea that, while an environment may behave arbitrarily, it is not its goal to prevent the system from satisfying its requirements. Hence, its behavior should not depend on the internal state of the system (which is encoded into the state space of the game). From the perspective of the environment player, we are thus dealing with a \emph{partial-information} setting. 
Simultaneously, from the system's perspective, the behavior of the environment can be arbitrary, but not \emph{arbitrarily complex}. 
We restrict the complexity of the environment's behavior by requiring that the environment uses only \emph{limited memory}.


Interestingly, partial-information games against limited memory opponents have not been studied in the past, even though they are a natural formalization for unknown environment behavior in reactive synthesis. In the present paper, we not only define this new class of games and discuss its basic properties, but also analyze the complexity of solving such games. In this first work on such games, we assume that the environment player does not obtain an additional stimulus from outside of the modeled scenario, and thus it only reacts to the system player's actions. We show that, surprisingly, even for memoryless environment players, checking if the system player has a winning strategy is PSPACE-complete for safety, reachability and parity objectives. Moreover, for a unary encoding of the memory, the problem can be solved in PSPACE. However, checking if the environment player has a winning strategy is NP-complete. This difference in the complexity highlights the asymmetry between the two players in this new game setting.
Note that for a binary encoding of the memory size, the upper bound on the amount of memory needed by a winning partial-information player can be concisely encoded, which means that the EXPTIME-hardness of incomplete-information game solving directly carries over to our setting \cite{Berwanger}.

We explore a new type of information asymmetry between the two players of a game: the environment player has limited information on the current state of the game, and the system player has no initial information on the goal of the environment player. The system player needs to learn from the environment's behavior \emph{while} avoiding situations in which the unknown behavior of the environment player could lead to the loss of the game for the system player \cite{Chakraborty2014Mar}. As such, we believe that our work prepares the ground for more formal approaches to the design of \emph{self-adaptive} systems that are able to achieve reasonable performance in completely unknown environments. In the future, we plan to research how to efficiently computes strategies for the system player in this new class of games.


\subsubsection*{Results, outline and related work.}

We study the problem of determining whether the partial-information player $P$ or the full-information player $F$ have a winning strategy in a game $G$. 
For $\mathcal{P} \in \{P, F\}$, let 
\begin{align*}
    \text{WIN}_\mathcal{P} = \{~ &\langle G \rangle~|~ G \text{ is a game arena and player } \mathcal{P} \text{ has a winning strategy}~\}.
\end{align*}

The games in the \emph{limited-memory} setting are \emph{not determined}, witnessed by the fact that the complexity of determining the existence of a winning strategy for a given game is different for each player as shown in Table \ref{tab:results}. The case in which both players have unlimited memory has been analyzed in \cite{Reif,kris,Berwanger} for different objectives. 
We study the WIN$_{\mathcal P}$ problem under reachability, safety, and parity objectives.
The complexity results are summarized in Table \ref{tab:results}. Additionally, we establish a lower bound on the memory required by the partial-information player to retain its winning strategies (an upper bound can be obtained from \cite{kris}). More specifically, we construct a game where the partial-information player needs at least $\mathcal{O}(\sqrt{n})$ bits of memory to preserve a winning strategy in a game of size $\mathcal{O}(n)$.  

\begin{table}
\centering
\caption{The complexity of determining if there exists a winning strategy.}
\label{tab:results}
\begin{tabular}{lll} 
\hline
Player $P$ & WIN$_P$  & WIN$_F$       \\ 
\hline
Memoryless   & NP-complete & PSPACE-complete  \\
Limited-memory    & NP-complete     & PSPACE-hard      \\
\hline
\end{tabular}
\end{table}

Next, we outline the main ideas in the case of a memoryless partial-information player with a reachability objective. 
To show the NP-completeness result of the WIN$_P$ problem,
we first observe that, given a strategy for the partial-information player, it is easy to check if it is winning. For proving hardness, we reduce the SAT problem to WIN$_P$.
It is harder for the full-information player to determine the existence of strategies. More specifically, this problem is PSPACE-complete. To prove that WIN$_F \in\mathit{PSPACE}$, we reduce an arbitrary instance of a game in our setting to a QBF formula with at most $\mathcal{O}(n^2)$ quantifier alternations, such that the full-information player has a winning strategy if and only if the formula is satisfiable.  For the hardness result, we define a game with an \emph{exit} action, which makes the game stop abruptly. If the exit action is played in the initial stage, the full-information player directly wins. If that happens after the initial stage, the partial-information player wins. Moreover, the partial-information player can play the exit action any step. Therefore, the full-information player has to ensure that\Dj{,} in the memoryless strategy played by the partial-information player, no state previously not visited is visited
after the initial stage, since the partial-information player can choose the exit action in such a case. 

We provide a reduction from QBF to this game, to prove PSPACE-hardness. Thus, we can encode the satisfaction of a QBF instance into the game graph. The two players initially build an assignment to the QBF variables sequentially, which the partial-information player can see. Afterward, the full-information player loses if the partial-information player takes the exit move. The only way for the full-information player to avoid losing is if it does not output a variable assignment part for which it has not previously observed that the partial-information player will not respond with exit.  By this manner, the same assignment has to be played repeatedly by both players. By taking the product of this game with an encoding of all clauses in a QBF formula, we obtain the required reduction from a QBF formula, as we know that the partial-information player will always choose the same values if the other player does, and the other player will need to do so to avoid an exit.

The idea behind the construction is that the full-information player cannot know that the assignment selected does not happen to be the only one not causing an exit; hence it has to assume that the partial-information player can play the exit move and hence is subsequently forced to solve the QBF problem to find an initial assignment for which it will not lose the game. 
Additionally, we show that the hardness carries over to all interesting objectives and also to the case when the partial-information player is forgetful.

For obtaining the lower bound on the memory required by the partial-information player to retain winning strategies, we use the idea of succinct counters from \cite{Berwanger} to encode prime numbers $\leq 2^n$ in a game graph of size $\mathcal{O}(n^2)$. The full-information player chooses a prime number randomly, and then the partial-information player has to guess this number correctly to win the game. When the partial-information player has unlimited memory, the constructed game is always winning for him. However, when we consider the limited memory scenario, the partial information player requires at least $n\#$ (primorial) distinct states in the memory.

\subparagraph*{Related work.}
In most straight forward setting of two-player games on graphs, the players are assumed to have perfect information about the state of the game. This setting has been thoroughly researched in \cite{Martin1975Sep,Emerson1993Jun,Thomas2002Jul,Henzinger2007Jan}. Most recently, quasi-polynomial algorithms for parity are given in \cite{Calude2017Jun,Lehtinen2018Jul}. 
The imperfect information games, where the players have asymmetric  knowledge about the state of the game has been studied in \cite{Reif,kris,Berwanger}. This setting is particularly useful in the synthesis of controllers that gather information about the state of the game using sensors.
In \cite{Chatterjee2005Dec} a class of semiperfect-information games, where one player has imperfect information and the other player has perfect information, is studied. This class is more straightforward than the games studied in \cite{kris}, and it can be solved in NP$\cap$coNP for parity objectives. However, in this setting, the ambiguity on the current state is not carried over.

\section{Definitions}

\subsubsection*{Arenas.}
We study two-player games in which each player chooses, in turn, a symbol from an alphabet. 
An infinite word is obtained as a result of the play. 
The game is played in an \emph{arena} $A = \langle V,\Sigma_P,\Sigma_F,\Delta, \init \rangle$ between a partial-information player $P$ and a full-information player $F$, where 
\begin{itemize}
    \item $V = V^{public} \times V^{private}$ is the state space;
    \item $\Sigma_P$ and $\Sigma_F$ are the sets of actions available to the respective players. Moreover, the actions of player $F$ have two components i) a public component $\Sigma_F^{public}$, and ii) a private component $\Sigma_F^{private}$, i.e., $\Sigma_F = \Sigma_F^{public} \times \Sigma_F^{private}$;
    \item $\Delta: V \times \Sigma_P \times \Sigma_F \to V$ is the transition function; and
    \item $\init$ is the initial state of the game.
\end{itemize}

\noindent The two players have asymmetric capabilities. The \emph{public component} of the state space $V^{public}$ is visible to both players. 
The \emph{private component} $V^{private}$ is not visible to player $P$. The actions of the full-information player are only partially visible to the partial-information player.

\subsubsection*{Strategies.}

In every step of the game Player $P$ first makes a move from $\Sigma_P$. Player $F$ responds to player $P$'s action by making a move from $\Sigma_F$. A strategy function for a player depends on the memory available to it and the information it has about the state of the game. If the next move of a player depends only the current state then the corresponding strategy is \emph{memoryless}. 
Memoryless strategies are a special case of limited-memory strategies. Formally, the strategy for the partial-information player $P$ is a function $S_P: V^{public} \times M \to \Sigma_P \times M$, where $M = \{1,\dots,k\}$. Player $P$ is memoryless if $k = 1$. The strategy $S_F$ for the full-information player is a function $S_F: V^+ \to \Sigma_F$.

\subsubsection*{Plays and Outcomes.}
A \emph{play} $\pi = \init (\sigma_1^P,\sigma_1^F)v_1(\sigma_2^P,\sigma_2^F)v_2 \dots$ is an alternating sequence of positions and action pairs, such that for all $i$, there is a valid transition between $v_i$ and $v_{i+1}$. 
The \emph{prefix} upto $v_n$ of $\pi$ is denoted $\pi(n)$. Its \emph{length} denoted $|\pi(n)|$ is $n$ and its \emph{last element} is $last(\pi(n)) = v_n$. The set of infinite plays in $A$ is denoted $\plays{A}$, and the set of corresponding finite prefixes is denoted $\pre{A}$.
The outcome of two strategies $S_P$ (for player $P$) and $S_F$ (for player $F$) in $A$ is the play $\pi = \init (\sigma_1^P,\sigma_1^F) v_1  (\sigma_2^P,\sigma_2^F)v_2 \dots \in \plays{A}$, such that for all $i$, 
a) $\sigma_i^P = S_P(v_{i-1},M)\downarrow1$ \footnote{$\downarrow i$ is the projection to the ith component},
b) $\sigma_i^F = S_F(v_0,\dots,v_{i-1},\sigma_i^P)$, and
c) $\Delta(v_i,(\sigma_i^P,\sigma_i^F))=v_{i+1}$.
This play is denoted $outcome(A,S_P,S_F)$.

\subsubsection*{Objectives.}

An \emph{objective} $\phi \subseteq V^{\omega}$ for arena $A$ is a set of infinite sequences of states. A play $\pi = \init (\sigma_1^P,\sigma_1^F) v_1 (\sigma_2^P,\sigma_2^F) v_2$ satisfies the objective $\phi$, denoted $\pi \models \phi$, if $\init v_1 v_2 \dots \in \phi$. We specifically consider reachability, safety, and parity objectives. For a play $\pi=\init (\sigma_1^P,\sigma_1^F) v_1 (\sigma_2^P,\sigma_2^F)\ldots$, we write $Inf(\pi)$ for the set of states that appear infinitely often in $\pi$. 
\begin{itemize}
    \item Reachability: Given a set $T \subseteq V$ of target states, the \emph{reachability} objective $Reach(T)$ requires that a state in $T$ is visited at least once, i.e., $Reach(T) = \{\init  v_1 v_2 \dots \in V^\omega~|~\exists k \geq 0: v_k \in T\}.$ 
    \item Safety: Dually, the \emph{safety} objective $Safe(T)$ requires that only states in $T$ are visited. Formally, $Safe(T) = \{~\init v_1 v_2 \dots \in V^\omega~|~ \forall k \in \mathbb{N}:  v_k \in T\} $.
    \item Parity: For $d \in \mathbb{N}$, let $p: V \to \{0,1,\dots,d\}$ be a priority function, which maps each state to a non negative integer priority. The \emph{parity objective} $Parity(p)$ requires that the minimum priority that appears infinitely often \Dj{is} even. Formally, $Parity(p) = \{\pi \in V^\omega|\min\{p(v)|v \in Inf(\pi)\} \text{ is even}\}$.
\end{itemize}

\subsubsection*{Games.}
A strategy $\alpha$ for a player with objective $\phi$ is \emph{winning} in $A$, if for every strategy $\beta$ of its opponent, if $\pi \in outcome(A,\alpha,\beta)$, then $\pi \models \phi$. A \emph{game} $G$ is a pair $(A, \phi)$, where $A$ is an arena as above, and $\phi$ is an objective. An objective $\phi$ for player $i$ is denoted $\phi_i$. 

\section{Memoryless partial-information player}

First, we consider the scenario when the partial-information player is memoryless. The results proved for a memoryless partial-information player are then extended to the forgetful case in the next section. We prove that the problem of determining the existence of winning strategies for the partial-information player is NP-complete irrespective of the objective.

\paragraph*{WIN$_P$ is NP-complete.}

Given a memoryless strategy for $P$, we can verify in polynomial time if it is winning. The following lemma is a direct consequence of this fact and is stated without proof.
\begin{lemma}
In a partial-information game with a memoryless partial-information player with parity objective, WIN$_P \in NP$. 
\label{winp:hard}
\end{lemma}

To show the NP-hardness of the WIN$_P$ problem, we only need to consider reachability and safety goals for player $P$. We first construct an arena corresponding to any boolean formula in $CNF$. Following this construction, we show how to encode this CNF formula using the acceptance conditions.

For any propositional formula $\phi = C_1 \land C_2 \land \ldots \land C_m$ in CNF over a set $V_1 = \{x_1,\dots,x_n\}$ of propositional variables, we construct an arena $A_\phi$ of size at most  $\mathcal{O}(mn)$. The game arena is constructed in a way such that Player $P$ wins if the game reaches a state corresponding to satisfaction of all the clauses. Simultaneously, the structure of the arena also ensures that the only way for player $P$ to ensure safety (defined later) is by producing a satisfying assignment for $\phi$.
Player $P$ chooses assignments for the propositional variables once for every clause. Since $P$ is memoryless, its assignments are the same for every clause. It has to find an assignment such that all the clauses are satisfied. The arena $A_\phi$ is presented in Figure \ref{fig:winp}.
Formally,
$A_\phi = \langle V, \Sigma_P, \Sigma_F, \Delta, \init \rangle$ is an arena of imperfect-information, where

\begin{itemize}
    \item 
    The set of states of the arena 
    $V = V^{public} \times V^{private}$, where
    
    $V^{public} = \{ \infty \} \cup \{\set{x}{1}{\bot}, \set{x}{1}{\top},\set{x}{2}{\bot},\set{x}{2}{\top},\dots,\set{x}{n}{\bot},\set{x}{n}{\top}\}$
    and
    $V^{private} = \{C_1,\ldots,C_m\} \times \{\top,\bot\}$.
    
    $V^{public}$ records the assignment to the propositional variables. The second component $V^{private}$ is used by player $F$ to check the satisfaction of any clause under a partial assignment revealed in the public state.
    
    \item $\Sigma_P = \{\top,\bot\}$, $\Sigma_F = \{d,n\}$. The actions of player $P$ correspond to assignments to the propositional variables. The action $n$ of player F corresponds to it proclaiming that a clause is not satisfied. Player $F$ loses the game if it incorrectly uses $n$. Otherwise, it has to use the action $d$.
    \item The initial state $\init = \left(\infty,(C_1,\bot) \right)$.
    \item The set of transistions $\Delta$ is as described below. In each of the transitions below, $a$ is either  $\top$ or $\bot$. \begin{enumerate}  
       
        \item $(\infty,(C_k,\bot)) \xrightarrow{(a,d)}  (\set{x}{1}{a},(C_k, \llbracket C_k \rrbracket_{x_1 = a} \footnote{$\llbracket C \rrbracket_{x = e}$ resolves to $\mathbf{true}$ if setting $x$ to $e$ in $C$ satisfies the clause, and resolves to $\mathbf{false}$ otherwise.}))$
        
        \item $(\set{x}{j}{*},(C_k,\alpha)) \xrightarrow{(e,d)} (\set{x}{j+1}{a},(C_k, \alpha \lor \llbracket C_k \rrbracket_{x_j = a}))$, where $(j < n-1)$
 
        \item $(\set{x}{n}{*},(C_k,\top)) \xrightarrow{(a,d)} (\infty, (C_{k+1},\bot))$, where $(k < m)$
        
        \item $(\set{x}{n}{*},(C_k,\bot)) \xrightarrow{(a,n)} (\set{x}{n}{a},(C_m,\bot))$, where $(k < m)$
        
        \item $(v,(C_m,\top)) \xrightarrow{(*,*)} (v, (C_{m},\top))$
        
        \item $(v,(C_m,\bot)) \xrightarrow{(*,*)} (v, (C_{m},\bot))$
    \end{enumerate}
\end{itemize}

\begin{figure*}[!htb]
    \centering
    \begin{tikzpicture}[->,>=stealth',shorten >=1pt,auto,node distance=2cm, semithick]
    \node[state, initial, fill = cyan, label = $\bot$] (C1S) {$\infty$};
    \node[below of = C1S] (C1S1) {};
    \node[below of = C1S1] (C1S2) {};
    \node[state, right of = C1S, label = $\llbracket C_1 \rrbracket_{x_1 = \top} $] (C1X1T) {$\set{x}{1}{\top}$};
    \node[state, below of = C1X1T, anchor = south ,label =below:$\llbracket C_1 \rrbracket_{x_1 = \bot} $] (C1X1F) {$\set{x}{1}{\bot}$};
    \node[state, right of = C1X1T, label = $\llbracket C_1 \rrbracket_{x_2 = \top} $] (C1X2T) {$\set{x}{2}{\top}$};
    \node[state, below of = C1X2T, anchor = south, label = $\llbracket C_1 \rrbracket_{x_2 = \bot} $] (C1X2F) {$\set{x}{2}{\bot}$};
    \node[right of  = C1X2T] (C1E1) {$\dots$};
    \node[right of  = C1X2F, anchor = south] (C1E2) {$\dots$};
    
    
    \node[state, right of = C1E1, fill = green]  (TC1XMT) {$\set{x}{n}{\top}$};
    \node[state, below of = TC1XMT, anchor = south, fill = green] (TC1XMF) {$\set{x}{n}{\bot}$};
    \node[inner sep = 5pt, fit = (TC1XMT)(TC1XMF),draw,loosely dashed,label =$\top$] (C1T) {};
    
    \node[state, right of = TC1XMT, fill = orange] (FC1XMT) {$\set{x}{n}{\top}$};
    \node[state, right of = TC1XMF, fill = orange] (FC1XMF) {$\set{x}{n}{\top}$};
    \node[inner sep = 5pt, fit = (FC1XMT)(FC1XMF),draw,loosely dashed,label =$\bot$] (C1F) {};
    
    \node[inner sep = 0.7cm,fit = (C1S)(FC1XMT)(FC1XMF),draw,loosely dashed] (C1) {};
    
    \path (C1S) edge node {$(\top,d)$} (C1X1T);
    \path [dotted] (C1S) edge node [anchor = east] {$(\bot,d)$} (C1X1F);
    \path (C1X1T) edge node {$(\top,d)$} (C1X2T); 
    \path [dotted] (C1X1F) edge node [anchor = north] {$(\bot,d)$} (C1X2F); 
    \path [dotted] (C1X1T) edge node [anchor = north] {$\bot$} (C1X2F);
    \path (C1X1F) edge node [anchor = south]{$\top$} (C1X2T);
    
    \node[below left] at (C1.north east) {$\mathbf{C_1}$};

    \node[state, below of = C1S1, fill = cyan, label = $\bot$] (C2S) {$\infty$};
    \node[below of = C2S] (C2S1) {};
    \node[below of = C2S1] (C2S2) {};
    \node[state, right of = C2S, label =below:$\llbracket C_2 \rrbracket_{x_1 = \top} $] (C2X1T) {$\set{x}{1}{\top}$};
    \node[state, below of = C2X1T, anchor = south ,label =below:$\llbracket C_2 \rrbracket_{x_1 = \bot} $] (C2X1F) {$\set{x}{1}{\bot}$};
    \node[state, right of = C2X1T, label = below:$\llbracket C_2 \rrbracket_{x_2 = \top} $] (C2X2T) {$\set{x}{2}{\top}$};
    \node[state, below of = C2X2T, anchor = south, label = below:$\llbracket C_2 \rrbracket_{x_2 = \bot} $] (C2X2F) {$\set{x}{2}{\bot}$};
    \node[right of  = C2X2T] (C2E1) {$\dots$};
    \node[right of  = C2X2F, anchor = south] (C2E2) {$\dots$};
    \node[below of = C2E2] (C2E3) {$\vdots$};

    \node[state, right of = C2E1, fill = green]  (TC2XMT) {$\set{x}{n}{\top}$};
    \node[state, below of = TC2XMT, anchor = south, fill = green] (TC2XMF) {$\set{x}{n}{\bot}$};
    \node[inner sep = 5pt, fit = (TC2XMT)(TC2XMF),draw,loosely dashed,label =$\top$] (C2T) {};
    
    \node[state, right of = TC2XMT, fill = orange] (FC2XMT) {$\set{x}{n}{\top}$};
    \node[state, right of = TC2XMF, fill = orange] (FC2XMF) {$\set{x}{n}{\top}$};
    \node[inner sep = 5pt, fit = (FC2XMT)(FC2XMF),draw,loosely dashed,label =$\bot$] (C2F) {};
    
    \node[inner sep = 0.7cm,fit = (C2S)(FC2XMT)(FC2XMF),draw,loosely dashed] (C2) {};
    
    \path (C2S) edge node [anchor = north] {$(\top,d)$} (C2X1T);
    \path [dotted] (C2S) edge node [anchor = east] {$(\bot,d)$} (C2X1F);
    \path (C2X1T) edge node {$(\top,d)$} (C2X2T); 
    \path [dotted] (C2X1F) edge node [anchor = north] {$(\bot,d)$} (C2X2F); 
    \path [dotted] (C2X1T) edge node [anchor = north] {$\bot$} (C2X2F);
    \path (C2X1F) edge node [anchor = south]{$\top$} (C2X2T);
    
    \path (C1T) edge node [anchor = west] {~~~~~$(*,*)$} (C2S);
    
    \node[below left] at (C2.north east) {$\mathbf{C_2}$};
    
    
    \node[state, below of = C2S2, fill = cyan, label = $\bot$] (CMS) {$\infty$};

    \node[state, right of = CMS, label =below:$\llbracket C_m \rrbracket_{x_1 = \top} $] (CMX1T) {$\set{x}{1}{\top}$};
    \node[state, below of = CMX1T, anchor = south ,label =below:$\llbracket C_m \rrbracket_{x_1 = \bot} $] (CMX1F) {$\set{x}{1}{\bot}$};
    \node[state, right of = CMX1T, label = below:$\llbracket C_m \rrbracket_{x_2 = \top} $] (CMX2T) {$\set{x}{2}{\top}$};
    \node[state, below of = CMX2T, anchor = south, label = below:$\llbracket C_m \rrbracket_{x_2 = \bot} $] (CMX2F) {$\set{x}{2}{\bot}$};
    \node[right of  = CMX2T] (CME1) {$\dots$};
    \node[right of  = CMX2F, anchor = south] (CME2) {$\dots$};
   
    \node[state, right of = CME1, fill = green]  (TCMXMT) {$\set{x}{n}{\top}$};
    \node[state, below of = TCMXMT, anchor = south, fill = green] (TCMXMF) {$\set{x}{n}{\bot}$};
    \node[inner sep = 5pt, fit = (TCMXMT)(TCMXMF),draw,loosely dashed,label =$\top$] (CMT) {};
    
    \node[state, right of = TCMXMT, fill = orange] (FCMXMT) {$\set{x}{n}{\top}$};
    \node[state, right of = TCMXMF, fill = orange] (FCMXMF) {$\set{x}{n}{\top}$};
    \node[inner sep = 5pt, fit = (FCMXMT)(FCMXMF),draw,loosely dashed,label =$\bot$] (CMF) {};
    
    \node[inner sep = 0.7cm,fit = (CMS)(FCMXMT)(FCMXMF),draw,loosely dashed] (CM) {};
    
    \path (CMS) edge node [anchor = north] {$(\top,d)$} (CMX1T);
    \path [dotted] (CMS) edge node [anchor = east] {$(\bot,d)$} (CMX1F);
    \path (CMX1T) edge node {$(\top,d)$} (CMX2T); 
    \path [dotted] (CMX1F) edge node [anchor = north] {$(\bot,d)$} (CMX2F); 
    \path [dotted] (CMX1T) edge node [anchor = north] {$\bot$} (CMX2F);
    \path (CMX1F) edge node [anchor = south]{$\top$} (CMX2T);
    
    \path (C2T) edge node [anchor = east] {$(*,*)$~~~} (CMS);
    
    \node[above left] at (CM.south east) {$\mathbf{C_m}$};

    \path [bend left] (C1F) edge node {$(*,n)$} (CMF.east) ;
    \path [bend left] (C2F) edge node [anchor = east]{$(*,n)$} (CMF.east) ;
    
    \path (CMF) edge [loop below] node {$(*,*)$} (CMF);
    \path (CMT) edge [loop below] node {$(*,*)$} (CMT);
    
    \end{tikzpicture}
    \caption{Games $G_\phi^{R}$ and $G_\phi^{S}$ are played on the arena $A_\phi$. $\phi$ is satisfiable if and only if a memoryless player $P$ has a winning strategy.}
    \label{fig:winp}
\end{figure*}
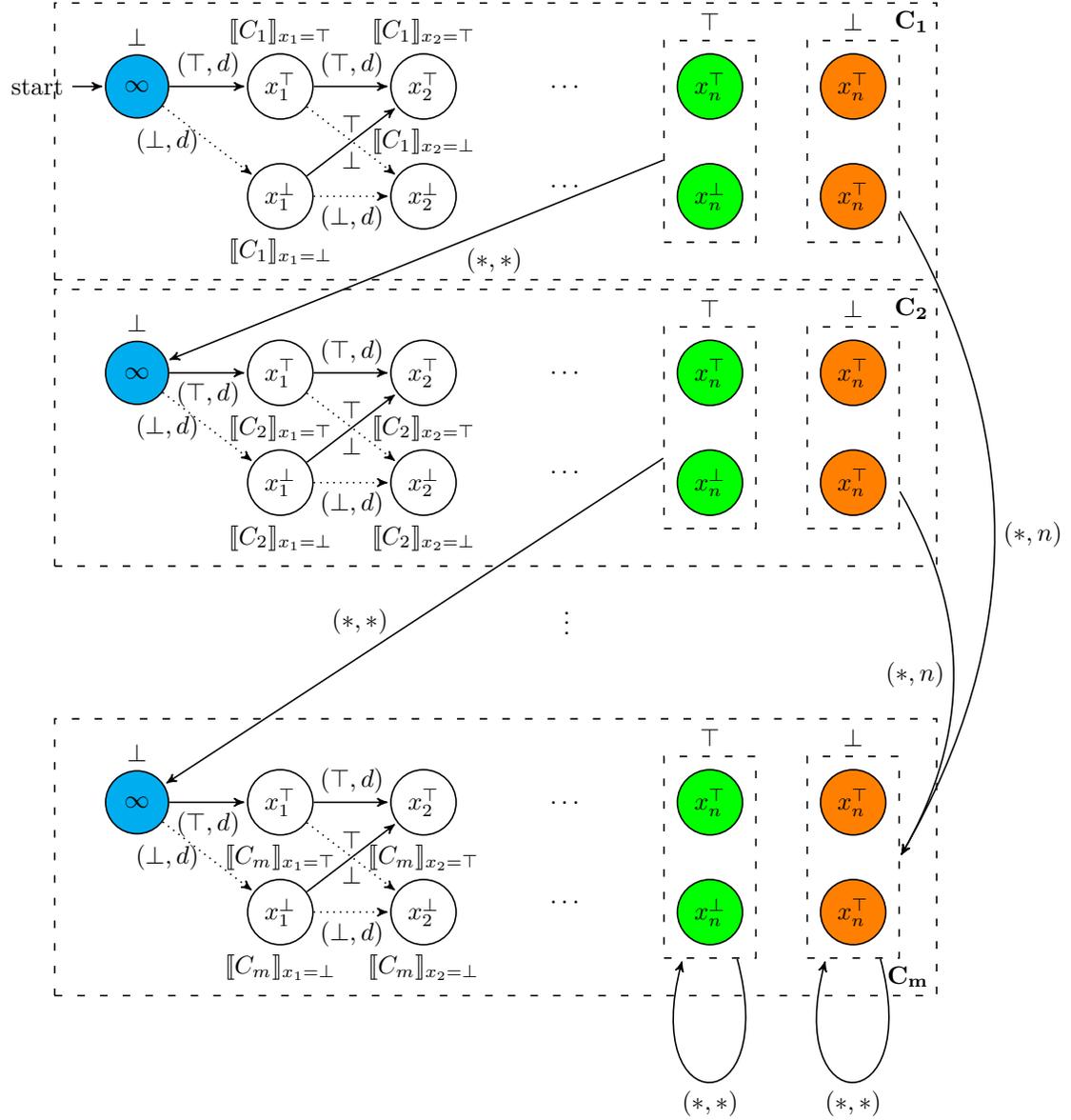
For proving the hardness of the reachability objective, we consider the reachability game $G_\phi^{R} = (A_\phi, Reach(R))$ ,where $R = \{ ( v,(C_m,\top))~|~v \in V^{public} \}$.  The following lemma is a direct consequence of the construction of the arena.
\begin{lemma}
$\phi$ is satisfiable if and only if memoryless player $P$ has a winning strategy for winning $G_\phi^{R} = (A_\phi,Reach(R))$.
\label{winp:reach}
\end{lemma}
\begin{proof}
See Appendix.
\end{proof}

\begin{corollary}
The WIN$_P$ problem with parity objective for player $P$ is NP-hard.
\label{corollary:winp}
\end{corollary}
To show that WIN$_P$ with safety objective for player $P$ is NP-hard, we consider the game $G_\phi^S = (A_\phi,Safe(S))$, where $S = V \setminus \{(x_n^\top,(C_m,\bot)),(x_n^\bot,(C_m,\bot))\}$. 

\begin{lemma}
$\phi$ is satisfiable if and only if memoryless player $P$ has a winning strategy for winning $G_\phi^{S}=(A_\phi,Safe(S))$.
\label{winp:safety}
\end{lemma}
\begin{proof}
See Appendix.
\end{proof}

\paragraph*{WIN$_F$ is PSPACE-complete.}

Player $F$ has more ways to defeat player $P$ when player $P$'s strategies are restricted. In this case we prove that determining if $F$ has a  winning strategy is harder. First we show that WIN$_F \in$ PSPACE. 
Consider a partial-information game $G$ played on arena $A$, with a parity objective for player $P$. We bound the maximum number of rounds needed by the full-information player $F$ to win this game. 

\begin{lemma}
If player $F$ can win this game, then he can win this game in $\mathcal{O}(n^2)$ rounds, where $n = |V|$.
\end{lemma}
\begin{proof}
We prove a more general lemma, and this lemma is a direct consequence.
\end{proof}
\begin{lemma}
Suppose $G = V_3 \coprod (V^{public} \times V^{private})$ 
such that player $P$'s moves on $V_3$ is fixed. Then, if player $F$ can win, it can win in $(|V_3| + |V^{public}||V^{private}|) (|V^{public}|+1)$ rounds.
\label{lemma:general_bound}
\end{lemma}
\begin{proof}
See Appendix.
\end{proof}

The next result is now direct. 

\begin{lemma}
In a partial-information game with a memoryless partial-information player with parity objective, WIN$_F \in$ PSPACE. 
\end{lemma}
\begin{proof}
Any two-player game which ends in a bounded number of rounds in the size of the game can be encoded as a QBF formula (formula game) \cite{Sipser:1996:ITC:524279}.
\end{proof}


Next, for any QBF formula $\psi$ of the form $ \forall x_1 \exists y_1 \forall x_2 \exists y_2 \dots \forall x_n \exists y_n. ~(C_1 \land C_2 \land \dots \land C_m),$ we construct an arena $A_\psi$ of size $\bigo(m\cdot n)$. The arena $A_\psi$ is presented in Figure \ref{fig:winf}.
Player $P$ has partial information and can use only positional strategies. The idea is that player $P$ has to assign the $x$ variables and player $F$ has to assign the $y$ variables such that the clauses $c_1,\dots,c_m$ are all satisfied. In each round the players have to make fresh assignments to all the variables. The rounds correspond to a fresh assignment for every clause. Both players have to ensure that its opponent makes the same assignment to its corresponding variables across different rounds (clauses). 

Player $P$ cannot assign different values as he can only play positional strategies. He has a exit move which he can use to push the game to a losing state for player $F$. By construction, he cannot use this move in the 1st round (1st clause) or else he will lose. 
As long as player $F$ makes the play come to the same public component when it comes to player $P$'s turn, player $P$ cannot use a exit move as he is forced to play positional strategies. But if player $F$ changes his assignment, then the visible part of the state will be something new (and player $F$ doesn't know player $P$'s move for this visible component), where player $P$ can play the exit move which will cause player $F$ to lose. This ensures that player $F$ cannot change the assignments as well. Now if there is a strategy for player $F$ to win this game, then the formula is satisfiable. If there is no strategy for player $F$ to win, then the formula is unsatisfiable.
\noindent Formally, $A_\psi =  \langle V, \Sigma_P, 
\Sigma_F, \Delta, \init \rangle$, where \begin{itemize}
    \item  The set of states of the arena $V = V^{public} \times V^{private}$, where \\
   $V^{public} = \{\infty\} 
                    ~\cup~\{\set{x}{1}{\bot},   \set{x}{1}{\top},\set{x}{2}{\bot},\set{x}{2}{\top},\dots,\set{x}{n}{\bot},\set{x}{n}{\top}\} 
                    ~\cup~\{\set{y}{1}{\bot}, \set{y}{1}{\top}, \set{y}{2}{\bot}, \set{y}{2}{\top}, \dots, \set{y}{n}{\bot}, \set{y}{n}{\top}\} \text{ and }$
    $V^{private} = \{C_1,\ldots,C_m\} \times \{\top,\bot\}$. $V^{public}$ records the assignment to the propositional variables that the 2 players jointly decide. The second component $V^{private}$ is used by player $F$(the full information player) to check the satisfaction of any clause under a partial assignment revealed in the visible state.
    \item $\Sigma_1 = \{\top,\bot,e\}$ and $\Sigma_2 = \{\top,\bot\}$. $e$ is a special move used by player $P$ (partial-information player) to proclaim that player $F$ (full-information player) has changed his assignments of variables after using a different assignment to satisfy some clause before. 
    \item The initial state $v_{init} = \left(\infty,(C_1,\bot) \right)$.
    \item The set of transitions $\Delta$ is as described below. In each of the transitions below, $a$ is either  $\top$ or $\bot$.
    \begin{enumerate}
       \item $(*,(C_1,*)) \xrightarrow{(e,*)} (\set{y}{n}{\bot},(C_m,\top))$
        
        \item $(\set{y}{j}{\top / \bot},(C_k,*)) \xrightarrow{(e,*)} (\set{y}{n}{\bot},(C_m,\bot))$ if $j < n$ and $m \geq k >1$

        \item $(\infty,(C_k,\alpha)) \xrightarrow{(a,*)}  (\set{x}{1}{a},(C_k, \alpha \lor \llbracket C_k \rrbracket_{x_1 = a}))$
        
         \item $(\set{x}{j}{*},(C_k,\bot)) \xrightarrow{(*,a)} (\set{y}{j}{a},(C_k),\llbracket C_k \rrbracket_{y_j = a})$
        
        \item $(\set{x}{j}{*},(C_k,\top)) \xrightarrow{(*,a)} (\set{y}{j}{a},(C_k, \top))$

        \item $(\set{y}{n}{*},(C_k,\top)) \xrightarrow{(*,*)} (\infty, (C_{k+1},\bot))$ where $(k < m)$
        
        \item $(\set{y}{n}{*},(C_k,\bot)) \xrightarrow{(*,*)} (\set{y}{n}{\bot},(C_m,\bot))$ where $(k < m)$
        
        \item $(\set{y}{j}{*},(C_k,\bot)) \xrightarrow{(a,*)} (\set{x}{j}{a},(C_k,\llbracket C_k \rrbracket_{y_j = a})) $ where $j < n$
        
        \item $(\set{y}{j}{*},(C_k,\top)) \xrightarrow{(a,*)} (\set{x}{j}{a},(C_k,\top))$ where $j < n$.
    \end{enumerate}
\end{itemize}

\begin{figure*}[!htb]
    \centering
    \begin{tikzpicture}[->,>=stealth',shorten >=1pt,auto,node distance=2cm, semithick]
    \node[state, initial, fill = cyan, label = $\bot$] (C1S) {$\infty$};
    \node[below of = C1S] (C1S1) {};
    \node[below of = C1S1] (C1S2) {};
    \node[state, right of = C1S, label = $\llbracket C_1 \rrbracket_{x_1 = \top} $] (C1X1T) {$\set{x}{1}{\top}$};
    \node[state, below of = C1X1T, anchor = south ,label =below:$\llbracket C_1 \rrbracket_{x_1 = \bot} $] (C1X1F) {$\set{x}{1}{\bot}$};
    \node[state, right of = C1X1T, label = $\llbracket C_1 \rrbracket_{y_1 = \top} $] (C1X2T) {$\set{y}{1}{\top}$};
    \node[state, below of = C1X2T, anchor = south, label = $\llbracket C_1 \rrbracket_{y_1 = \bot} $] (C1X2F) {$\set{y}{1}{\bot}$};
    \node[right of  = C1X2T] (C1E1) {$\dots$};
    \node[right of  = C1X2F, anchor = south] (C1E2) {$\dots$};
    
    
    \node[state, right of = C1E1, fill = green]  (TC1XMT) {$\set{y}{n}{\top}$};
    \node[state, below of = TC1XMT, anchor = south, fill = green] (TC1XMF) {$\set{y}{n}{\bot}$};
    \node[inner sep = 5pt, fit = (TC1XMT)(TC1XMF),draw,loosely dashed,label =$\top$] (C1T) {};
    
    \node[state, right of = TC1XMT, fill = orange] (FC1XMT) {$\set{y}{n}{\top}$};
    \node[state, right of = TC1XMF, fill = orange] (FC1XMF) {$\set{y}{n}{\top}$};
    \node[inner sep = 5pt, fit = (FC1XMT)(FC1XMF),draw,loosely dashed,label =$\bot$] (C1F) {};
    
    \node[inner sep = 0.7cm,fit = (C1S)(FC1XMT)(FC1XMF),draw,loosely dashed] (C1) {};

    
    
    \path (C1S) edge node {$(\top,*)$} (C1X1T);
    \path [dotted] (C1S) edge node [anchor = east] {$(\bot,*)$} (C1X1F);
    \path (C1X1T) edge node {$(*,\top)$} (C1X2T); 
    \path [dotted] (C1X1F) edge node [anchor = north] {$(*,\bot)$} (C1X2F); 
    \path [dotted] (C1X1T) edge node [anchor = north] {$\bot$} (C1X2F);
    \path (C1X1F) edge node [anchor = south]{$\top$} (C1X2T);
    \path[bend left] (C1.east) edge node [anchor =  north west]{$(e,*)$}(CMT);
    \path (CMT) [loop below] edge node {$(*,*)$} (CMT);
    
    \node[below left] at (C1.north east) {$\mathbf{C_1}$};

    \node[state, below of = C1S1, fill = cyan, label = $\bot$] (C2S) {$\infty$};
    \node[below of = C2S] (C2S1) {};
    \node[below of = C2S1] (C2S2) {};
    \node[state, right of = C2S, label =below:$\llbracket C_2 \rrbracket_{x_1 = \top} $] (C2X1T) {$\set{x}{1}{\top}$};
    \node[state, below of = C2X1T, anchor = south ,label =below:$\llbracket C_2 \rrbracket_{x_1 = \bot} $] (C2X1F) {$\set{x}{1}{\bot}$};
    \node[state, right of = C2X1T, label = below:$\llbracket C_2 \rrbracket_{y_1 = \top} $] (C2X2T) {$\set{y}{1}{\top}$};
    \node[state, below of = C2X2T, anchor = south, label = below:$\llbracket C_2 \rrbracket_{y_1 = \bot} $] (C2X2F) {$\set{y}{1}{\bot}$};
    \node[right of  = C2X2T] (C2E1) {$\dots$};
    \node[right of  = C2X2F, anchor = south] (C2E2) {$\dots$};
    \node[below of = C2E2] (C2E3) {$\vdots$};

    \node[state, right of = C2E1, fill = green]  (TC2XMT) {$\set{y}{n}{\top}$};
    \node[state, below of = TC2XMT, anchor = south, fill = green] (TC2XMF) {$\set{y}{n}{\bot}$};
    \node[inner sep = 5pt, fit = (TC2XMT)(TC2XMF),draw,loosely dashed,label =$\top$] (C2T) {};
    
    \node[state, right of = TC2XMT, fill = orange] (FC2XMT) {$\set{y}{n}{\top}$};
    \node[state, right of = TC2XMF, fill = orange] (FC2XMF) {$\set{y}{n}{\top}$};
    \node[inner sep = 5pt, fit = (FC2XMT)(FC2XMF),draw,loosely dashed,label =$\bot$] (C2F) {};
    
    \node[inner sep = 0.7cm,fit = (C2S)(FC2XMT)(FC2XMF),draw,loosely dashed] (C2) {};
    
    \path (C2S) edge node [anchor = north] {$(\top,*)$} (C2X1T);
    \path [dotted] (C2S) edge node [anchor = east] {$(\bot,*)$} (C2X1F);
    \path (C2X1T) edge node {$(*,\top)$} (C2X2T); 
    \path [dotted] (C2X1F) edge node [anchor = north] {$(*,\bot)$} (C2X2F); 
    \path [dotted] (C2X1T) edge node [anchor = north] {$\bot$} (C2X2F);
    \path (C2X1F) edge node [anchor = south]{$\top$} (C2X2T);
    
    \path (C1T) edge node [anchor = west] {~~~~~$(*,*)$} (C2S);
    
    \node[below left] at (C2.north east) {$\mathbf{C_2}$};
    
    
    \node[state, below of = C2S2, fill = cyan, label = $\bot$] (CMS) {$\infty$};

    \node[state, right of = CMS, label =below:$\llbracket C_m \rrbracket_{x_1 = \top} $] (CMX1T) {$\set{x}{1}{\top}$};
    \node[state, below of = CMX1T, anchor = south ,label =below:$\llbracket C_m \rrbracket_{x_1 = \bot} $] (CMX1F) {$\set{x}{1}{\bot}$};
    \node[state, right of = CMX1T, label = below:$\llbracket C_m \rrbracket_{x_2 = \top} $] (CMX2T) {$\set{x}{2}{\top}$};
    \node[state, below of = CMX2T, anchor = south, label = below:$\llbracket C_m \rrbracket_{x_2 = \bot} $] (CMX2F) {$\set{x}{2}{\bot}$};
    \node[right of  = CMX2T] (CME1) {$\dots$};
    \node[right of  = CMX2F, anchor = south] (CME2) {$\dots$};
   
    \node[state, right of = CME1, fill = green]  (TCMXMT) {$\set{y}{n}{\top}$};
    \node[state, below of = TCMXMT, anchor = south, fill = green] (TCMXMF) {$\set{y}{n}{\bot}$};
    \node[inner sep = 5pt, fit = (TCMXMT)(TCMXMF),draw,loosely dashed,label =$\top$] (CMT) {};
    
    \node[state, right of = TCMXMT, fill = orange] (FCMXMT) {$\set{y}{n}{\top}$};
    \node[state, right of = TCMXMF, fill = orange] (FCMXMF) {$\set{y}{n}{\top}$};
    \node[inner sep = 5pt, fit = (FCMXMT)(FCMXMF),draw,loosely dashed,label =$\bot$] (CMF) {};
    
    \node[inner sep = 0.7cm,fit = (CMS)(FCMXMT)(FCMXMF),draw,loosely dashed] (CM) {};
    
    \path [bend left] (CM.east) edge node {~~$(e,*)$} (CMF);
    
    \path (CMS) edge node [anchor = north] {$(\top,d)$} (CMX1T);
    \path [dotted] (CMS) edge node [anchor = east] {$(\bot,d)$} (CMX1F);
    \path (CMX1T) edge node {$(\top,d)$} (CMX2T); 
    \path [dotted] (CMX1F) edge node [anchor = north] {$(d,\bot)$} (CMX2F); 
    \path [dotted] (CMX1T) edge node [anchor = north] {$\bot$} (CMX2F);
    \path (CMX1F) edge node [anchor = south]{$\top$} (CMX2T);
    
    \path (C2T) edge node [anchor = east] {$(*,*)$~~~} (CMS);
    
    \node[below left] at (CM.north east) {$\mathbf{C_m}$};
    
    \path (C2) edge node [anchor = west]{$(e,*)$} (CMF);
    \path (CMF) edge [loop below] node {$(*,*)$} (CMF);
    \path (CMT) edge [loop below] node {$(*,*)$} (CMT);
    
    \end{tikzpicture}
    \caption{Games $G_\psi^{R}$ and $G_\psi^{S}$ are played on the arena $A_\psi$. $\psi$ is satisfiable if and only if a memoryless player $P$ has a winning strategy. If player $P$ plays an $e$ move in $C_1$, then it loses the game.} 
    \label{fig:winf}
\end{figure*}
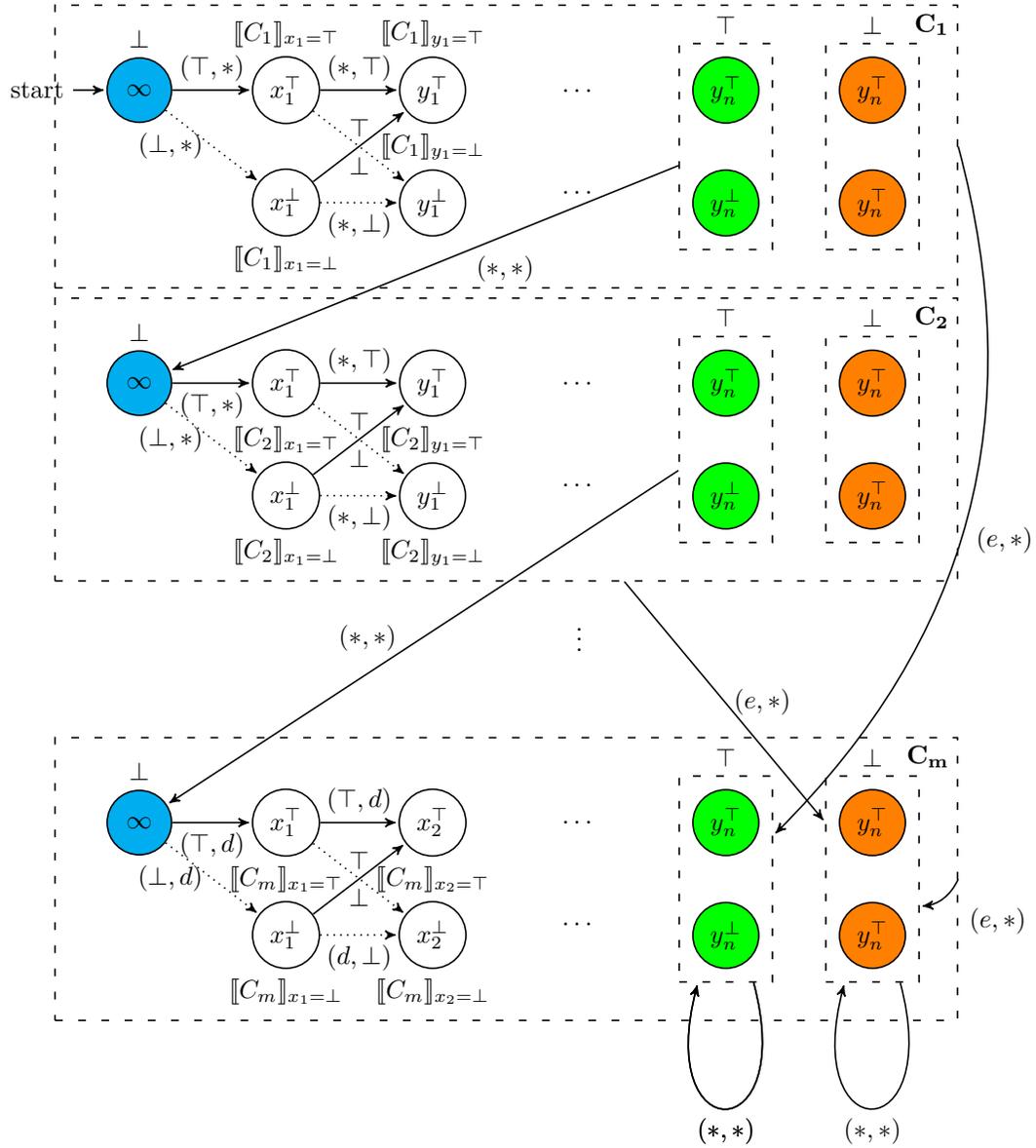
To prove PSPACE-hardness of the WIN$_F$ problem for different objectives, we only need to consider reachability and safety objectives. 
For safety objective for $P$, we consider the game $G_\psi^{S} = (A_\phi, Safe(V \setminus R))$, where $R = \{\left(v,(C_m,\top) \right) | v \in V_1 \}$. Player $F$ wins if the game reaches a state corresponding to satisfaction of all the clauses. The following lemma is a direct consequence of the definition of the game.
\begin{lemma}
$\psi$ is satisfiable if and only if player $F$ has a winning strategy for winning $G_\psi^{S} = (A_\psi,~safe(V \setminus R))$.
\label{lemma:winf_safe}
\end{lemma}

\begin{proof}
See Appendix.
\end{proof}

\begin{corollary}
The WIN$_F$ problem with parity objective for player $P$ is PSPACE-hard.
\end{corollary}

To show that WIN$_F$ with reachability objective for player $P$ is PSPACE-hard, we consider the game $G_\psi^R = (A_\phi,Reach(R))$, where $R =  \{(y_n^\top,(C_m,\bot)),(y_n^\bot,(C_m,\bot))\}$. The proof of the next lemma is similar to the previous one.

\begin{lemma}
$\psi$ is satisfiable if and only if player $F$ has a winning strategy for winning $G_\psi^{R} = (A_\psi,~Reach(R))$.
\label{lemma:winf_reach}
\end{lemma}
\begin{proof}
The argument from the proof of Lemma \ref{lemma:winf_safe}.
\end{proof}

\section{Limited-memory partial-information player}

In this section, the partial-information player $P$ can use a memory object $M = \{1,2,\ldots,k\}$ for his strategies. Therefore, its strategy $\sigma_P$ is a function $\sigma_P: V_{public} \times M \to \Sigma_P \times M$. Consider a game $G = (A,~Parity(p))$ played on the arena $A = \langle V, \Sigma_P, \Sigma_F, \Delta, \init \rangle$. We construct an equivalent game $G^k = (A^k,~Parity(p^k))$ played between a memoryless player $P$ and player $F$. In the arena $A^k$, all the uses of memory by player $P$ is encoded in the states. The main idea is that player $F$ forces player $P$ to exhaust all possible uses of memory for every public component, further this can be encoded in the game.
Formally $A^k = \langle V',\Sigma_P',\Sigma_F', \Delta', \init' \rangle$ where,
\begin{itemize}
    \item  The set of states of the arena $V' = (V^{public}) \times (V^{private}) \times M$. 
    
    \item $\Sigma_P' = \Sigma_P \times M$ and $\Sigma_F' = \Sigma_F$.
    \item The initial state $\init' = (v_0^{public},v_0^{private},k)$.
    
    \item The set of transitions $\Delta'$ is,
    \item $(u,v,i) \xrightarrow{((x,j),y)} (u',v,j)$, 
    where $\Delta \left( (u,v),(x,y) \right) = (u',v')$.
\end{itemize}

We only consider parity objectives for the partial information player, as they can encode both reachabilty and safety. Let $Parity(p)$ be the parity objective of the partial-information player. 
$p^k\big((u,v,i)\big) = p\big(\left(u,v\right)\big)$ is the priority function $p^k$ for the new arena. In the new arena, the priority of a state is the corresponding priority in the original arena. The following lemma is straight forward.

\begin{lemma}
A bounded-memory player $P$ has a winning strategy against player $F$ in the game $G$ if and only if a memoryless $P$ has a winning strategy against player $F$ in the game $G^k$ defined above.
\end{lemma}  

Now as a consequence of Corollary \ref{corollary:winp} the existence of a winning strategy for player $P$ is also NP-complete.
\begin{corollary}
If player $P$ can use bounded memory, WIN$_P$ is NP-complete.
\end{corollary}

Next, we show that WIN$_F$ remains PSPACE-hard even when player $P$ can use bounded memory. For every game $G = (A,Parity(p))$ played between a bounded-memory player $P$ and $F$, we construct an equivalent game $G_k = (A_k,Parity(p_k))$ played between a bounded-memory player $P$ and $F$.  In the game $G_k$, player $P$ has $k$ unique memory values, therefore for every $u \in V_{public}$ it can make at most $k$ different choices from its actions. However, it cannot make different choices since the game will reach a winning sink for player $F$. The arena $A_k = \langle V_k,\Sigma_P',\Sigma_F',\Delta',v_0' \rangle$ where,
\begin{itemize}
    \item The set of states $V_k = V^{public} \times (V^{private})' \cup \{W\}$, where $(V^{private})' = V^{private} \times C \times R$. $C = \{0,\dots,k\}$. It is a counter used by player $F$ to exhaust all the options of memory that player $P$ can use at any state. Further $R = \Sigma_P$. Player $F$ uses $R$ to record the move used by player $P$ when the counter $C$ has the value $0$. $W$ is a state that is winning for player $F$, once the game reaches $W$ 
    \item The initial state $v_0'= (u_0,(v_0,0,\bar{x}))$, where $(u_0,v_0)$ is the initial state of $A$ and $\bar{x}$ is some fixed move in $R$.
    \item The actions $\Sigma_P' = \Sigma_P \times \{1,\dots,k\}$ and $\Sigma_F' = \Sigma_F$.
    \item In the below set of transitions, $\Delta((u,v),(x,y)) = (u',v')$. The set of transitions $\Delta'$ is,
    \begin{enumerate}
        \item $(u,(v,j,x)) \xrightarrow{((x,j),y)} (u,(v,j+1,x))$, if $1 \leq j < k-1$.
        \item $(u,(v,0,*)) \xrightarrow{((x,1),y)} (u,(v,1,x))$.
        \item $(u,(v,k-1,x)) \xrightarrow{((x,0),y)} (u',(v',0,\bar{x}))$.
        \item $(u,(v,j,x)) \xrightarrow{(e,*)}  W$ if $1 \leq j<k-1$ and $e \neq (x,j+1)$
        \item $(u,(v,k-1,x)) \xrightarrow{(e,*)}  W$ if $e \neq (x,0)$
    \end{enumerate}
    \item The priority function for the arena $p_k(u,(v,*,*)) = p(u,v)$ and $p_k(W) = 1$.
\end{itemize}

\begin{lemma}
In the game $G_k$, player $F$ has a winning strategy if and only if it also has winning strategy in $G$ against a memoryless $P$.
\end{lemma}

\begin{corollary}
WIN$_F$ is PSPACE-hard.
\end{corollary}


\section{Lower bound on memory for the partial-information player} 
In this section, we prove a lower bound on the memory required for the partial-information player to retain winning strategies. As the order of play does not matter without loss of generality, we will assume that the full-information player $F$ plays first. In this case, strategy for player $F$ is a function $\pi: V^+ \to \Sigma_F$. Strategy for player $P$ is a function $\sigma: V^{public} \times M \times \Sigma_F \to \Sigma_P \times M$. $M$ is the memory object available to player $P$ and $|M| = k$. $p_n$ denotes the $n^{th}$ prime number. 

\noindent \underline{Remark}: $p_n < n^2$.

\subsection*{The prime-remainder game.}
Player $P$ and player $F$ play a game in two stages. In the first stage, player $F$ chooses a prime number number $p \in \{p_1,\dots, p_n\},$ by playing an invisible move $p$ and number 
$N \in \{0,1,\dots,\prod_i^n p_i\}.$ A number $N$ can be chosen by player $F$ by making the move $s$ $N$ times. After this player $F$ reveals the prime number $p$ to player $P$ by playing the move $p$ publicly. The remainder $r = N~mod~p$ is stored in the invisible component $V^{private}$. In the second stage, player $P$ has to guess the remainder $r$ correctly by playing the move $f$ correct number of times then playing an $s$. The arena is presented in Figure \ref{fig:lowerbound1}. Player $F$ wins the game if the game reaches a state (corresponding to player $P$ incorrectly guessing the remainder) in $F = \{ (S,v_2)~|~v_2 \text{ is colored blue}\} \cup \{ (F,X)\}.$

\begin{figure*}[!htb]
\begin{subfigure}{.3\textwidth}
\begin{tikzpicture}[->,>=stealth',shorten >=1pt,auto,node distance=2.2cm, semithick]

\node[state, initial, fill = cyan] (S) {$S$};
\node[state, fill = orange] (F) [below of = S] {$F$};
  
\path (S) edge [loop above] node {$(s,*)$} (S);
\path (S) edge node [anchor = east] {$(p,*)$} (F);
\path (F) edge [bend right] node [anchor = west] {$(*,s)$} (S);
\path (F) edge [loop below] node {$(*,f)$} (F);  
  
\node[inner sep = 1.6cm,fit = (S)(F),draw,loosely dashed] (vis){};
\node[above left] at (vis.south east) {$\mathbf{V^{public}}$}; 
\end{tikzpicture}
\end{subfigure}
\begin{subfigure}{.2\textwidth}
\begin{tikzpicture}[->,>=stealth',shorten >=1pt,scale=0.6,node distance=1.6cm, semithick]

\node[initial,state] (N50){$0$};
\node[right of = N50] (N) {};
\node[state,right of = N] (N51) {$1$};
\node[right of = N51] (E) {$\ldots$};
\node[state,right of = E] (N55) {$p_n-1$};

\path (N50) edge node [anchor = south]{$(s,*)$} (N51);
\path (N55) edge [bend right] node [anchor = north] {$(s,*)$} (N50);

\node[below of = N50] (E1) {};

\node[state, below of = E1, fill = orange] (B) {$X$};
\node[state, right of = B] (RP0) {$0$};
\node[state, right of = RP0, fill = cyan] (RP1) {$1$};
\node[right of = RP1] (RP2) {$\ldots$};
\node[state, right of = RP2, fill = cyan] (RP5) {$p_n-1$};

\path (RP0) edge node [anchor = north] {$(*,f)$} (B);
\path (RP1) edge node [anchor = north] {$(*,f)$} (RP0);

\path (N50) edge node [anchor = west] {$(p,*)$} (RP0);
\path (N51) edge node [anchor = west] {$(p,*)$} (RP1);
\path (N55) edge node [anchor = west] {$(p,*)$} (RP5);

\path (RP0) edge [loop below] node [anchor = north] {$(*,s)$} (RP0);
\path (RP1) edge [loop below] node [anchor = north] {$(*,s)$} (RP1);
\path (RP5) edge [loop below] node [anchor = north] {$(*,s)$} (RP5);

\node [below of = B] (B1) {};
\node [below of = B1, state] (P1) {$p_1$};
\node [right of = P1, state] (P2) {$p_2$};
\node [right of = P2, state] (P3) {$p_3$};
\node [right of = P3] (PE) {$\dots$};
\node [right of = PE, state] (PN) {$p_n$};

\node [draw, fit = (P1)(PN), inner sep = 2pt,loosely dashed] (PRIMES) {};

\node[inner sep = 1.3cm, fit = (N50)(RP5),draw,loosely dashed](invis1){};

\node [draw, fit = (PRIMES)(invis1), inner sep = 10pt,loosely dashed] (invis) {};
\node[above left] at (invis.south east) {$\mathbf{V^{private}}$};
\end{tikzpicture}
\end{subfigure}

\caption{Prime-remainder game. Player $F$ can choose a prime number $p \in \{p_1 \dots p_n \}$ by playing an invisible move. It reveals the number $N$ by playing $s$ $N$ times. The remainder $N~mod~p$ is stored in the invisible component. Then prime number $p$ is revealed by player $F$. Now player $P$ has to guess the remainder correctly by playing $f$ suitable number of times. Player $F$ wins the game, if player $P$ guesses the number incorrectly.}
\label{fig:lowerbound1}
\end{figure*}
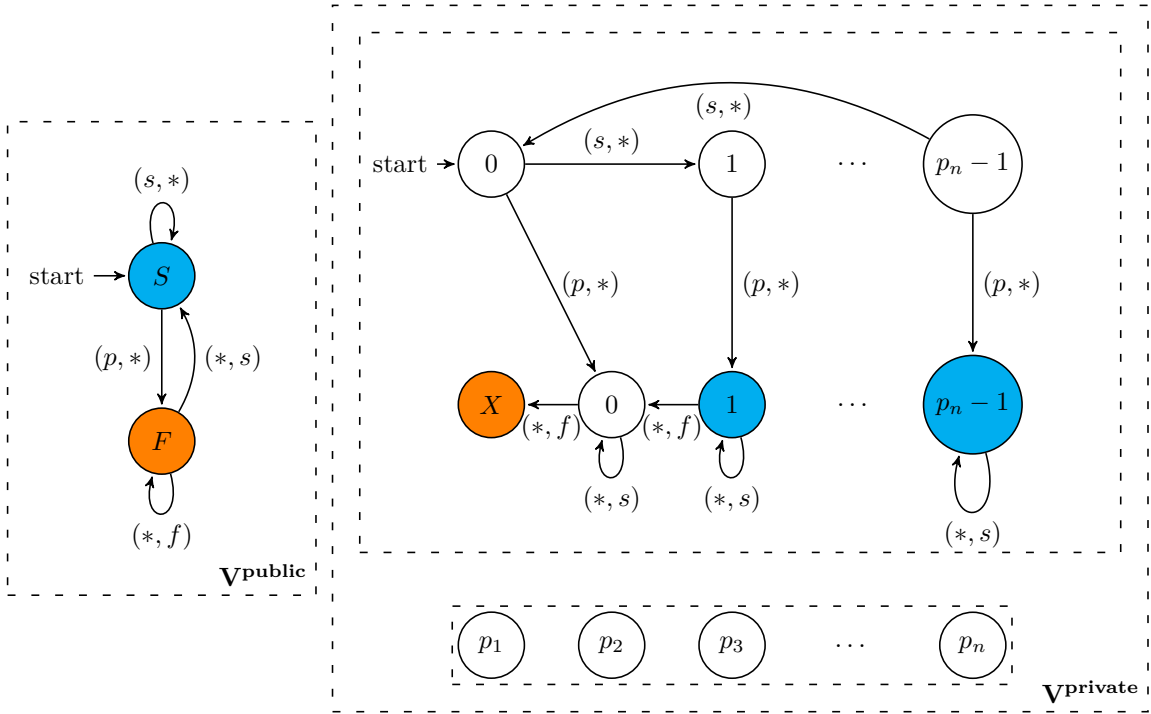

\subparagraph*{Analysis of the prime-remainder game.}

\begin{theorem}
Player $P$ always has a winning strategy when it can use unlimited memory.
\end{theorem}
\begin{proof}
Player $P$ can record the number $N$ itself by using a memory structure with $\prod_i^n p_i$ distinct states. Then player $P$ has a winning strategy against all strategies of player $F$.
\end{proof}

\noindent \underline{Remark:} $n\# = \prod_i^n p_i = e^{(1+o(1))n}$.

\begin{lemma}
Player $P$ needs a memory object with at least that $n\#$ distinct states to win.
\end{lemma}
\begin{proof}
The Chinese-remainder theorem implies that $N$ is uniquely determined by its remainders against the set of primes $\{p_1,p_2,\dots ,p_n \}$. Player $F$ can choose any prime in $p_1, \dots ,p_n$ and this prime is revealed only after the number has been chosen.  Since player $P$ cannot know apriori the choice of the prime number $p$, he has to record all the remainders.  Therefore, he requires at least $(p_1 - 1)(p_2 - 1)\dots(p_n - 1)$ states.
\end{proof}

\begin{corollary}
To record $(p_1 - 1)(p_2 - 1)\dots(p_n - 1)$ states we need $\mathcal{O}(n)$ bits of memory. This game has $\mathcal{O}(n^2)$ states as $p_n < n^2$. To win in a game of size $\mathcal{O}(n)$, we need a memory with $\mathcal{O}(\sqrt{n})$ bits of memory.
\end{corollary}

\section{Conclusion}

To our best knowledge, this paper offers the first study where \emph{limited memory} is used to capture the non-adversarial nature of the opponent. Moreover, the opponent is also \emph{forgetful}.
We show that the existence of winning strategies for the partial-information player is an NP-complete problem. On the other hand, the existence of winning strategies for the full-information player is a PSPACE-hard problem. Such games played against a weak opponent are asymmetric for the two players, as observed from the complexity results. As expected, limiting the memory available to the opponent introduces new winning strategies. However, these strategies are harder to find when compared to the traditional setting. Surprisingly, the worst-case complexity for determining the existence of winning strategies is independent of the type of the winning condition. 
However, it is not clear if it is PSPACE-complete.


\clearpage
\nocite{Gradel2002}
\nocite{Sipser:1996:ITC:524279}
\nocite{Thomas2002Jul}
\nocite{Martin1975Sep}
\nocite{Henzinger2007Jan}
\bibliography{ref}

\clearpage

\appendix

\section{Proofs}

\begin{claimproof}[Proof of Lemma \ref{winp:hard}]
($\Rightarrow$) If $\phi$ is satisfiable, then there is an assignment to all the variables that satisfies every clause, player $P$ simply chooses this assignment. Player $F$ can never play the $n$-move as it will take him to a losing state.
\\
($\Leftarrow$) If $\phi$ is unsatisfiable, then no matter the assignment chosen by player $P$, there is a clause $k$ that is not satisfied. Player $F$ can play the $n$-move in this clause to win the game.
\end{claimproof}

\begin{claimproof}[Proof of Lemma \ref{winp:safety}]
When the formula $\phi$ is not satisfiable, the game reaches a state of the form $(*,(C_k,\bot))$. In this state, player $F$ can play the $n$ move to defeat player $P$. The next lemma formalizes this.
\end{claimproof}

\begin{claimproof}[Proof of Lemma \ref{lemma:general_bound}]
Proof is by induction on $|V^{public} \times V^{private}|$. Suppose $V^{public}  = V^{private} = \emptyset$, then all of player $P$'s moves are fixed. Hence if player $P$ can win, it can win in $|V_3|$ moves. The start state of the game can be in either 
\begin{enumerate}
    \item{$V^{public} \times V^{private}$:} Say the game starts in $(v_1,v_2)$. After the first round player $P$'s moves with respect to $\{v_1\} \times V^{private}$ are fixed for the rest of the game. Hence, we set $V_3 = V_3 \cup (\{v_1\}\times V^{private})$ and $V^{public} = V^{public} \setminus \{v_1\}$. Thus, player $F$ wins in $1 + (|V_3| + |V^{private}| + (|V^{public}| - 1)|V^{private}|)(|V^{public}| - 1 + 1) \leq (|V_3| + |V^{public}||V^{private}|) (|V^{public}|+1)$.
    \item $V_3$: If player $F$ can win the game by staying inside $V_3$, then it can win the game in $|V_3|$ moves. Otherwise player $F$ has to escape the $V_3$ sub-graph in at most $|V_3|$ moves as player $P$'s moves in $V_3$ are known, else it remains stuck in $V_3$. Say it escapes to $(v_1,v_2)$, using the previous case player $F$ wins in $\leq |V_3| + (|V_3| + 1 + (|V_3| + |V^{private}| + (|V^{public}| - 1)|V^{private}|)(|V^{public}| - 1 + 1) \leq (|V_3| + |V^{public}||V^{private}|) (|V^{public}|+1)$ moves. 
\end{enumerate}
\end{claimproof}

\begin{claimproof}[Proof of Lemma \ref{lemma:winf_safe}]
($\Rightarrow$) 
If the formula $\psi$ is satisfiable, for any assignment to a variable $x_i$ by player $P$, player $F$ has a suitable response. The model for the formula is also a winning strategy for player $F$. Observe that the partial-information player cannot perform the exit move $c$ as this will directly take him to a winning state of player $F$ (by construction rules 1 and 2 of the transition).
\\
($\Leftarrow$) We prove the contrapositive. We show that if $\psi$ is not satisfiable, then player $F$ does not have a winning strategy. Since $\psi$ is not satisfiable, the player $P$ has a winning strategy in the standard game interpretation. This implies in any path traversed in the decision tree, there is a clause in $\psi$ that is not satisfied. Let $p$ be one such path. 

Suppose if $C_1$ evaluates to $\bot$ in $p$, then some state $(\set{y}{n}{*},(C_1,\bot))$ is reached. Further by rule $7$ this always leads to a state where there is no transition out for player $F$ (it does not have a path to any winning state). If this is not the case, there is some clause $C_k$ that is not satisfiable under $p$. If player $F$ plays the same assignments for every clause under $p$, by the same argument above it would reach a dead end. Hence player $F$ has to necessarily make a different choice for some $\exists$ variable $y_j$. Let $c_l$ $(1< l \leq k)$ be the clause where it deviates from $p$ for the 1st time. If player $P$ responds by making the exit move (Note that it can do this as per a memoryless strategy since it is the first time it is encountering this vertex in the visible component of the game) irrespective of player $F$'s response the game hits a dead end as per rule 2. 
\end{claimproof}

\end{document}